\documentclass[journal]{IEEEtran}
\usepackage{color}
\usepackage{graphicx}
\usepackage{algorithm}
\usepackage{algorithmic}
\usepackage[small,bf]{caption}
\usepackage{epsfig}
\usepackage{epstopdf}  %{--nocompress --gsopt="-dPDFSETTINGS=/prepress -dSubsetFonts=true -dEmbedAllFonts=true" input.eps}
\usepackage{amsfonts}
\usepackage{latexsym,amssymb,amsmath}
\newtheorem{theorem}{Theorem}
\newtheorem{lemma}{Lemma}
\newtheorem{corollary}{Corollary}

\newenvironment{proof}{\noindent{\bf Proof.}}{\hfill \qed \vskip 5pt}
\def\qed{\hfill\rule{2mm}{2mm}}

\begin{document}
\title{Power Strip Packing of Malleable Demands in Smart Grid}
%\author{Mohammad M. Karbasioun, Gennady Shaikhet, Evangelos Kranakis, Ioannis Lambadaris\\
%			Carleton University, Ottawa, ON, Canada\\
%			
%			E-mails: \{mkarbasi@sce, gennady@math, kranakis@scs, Ioannis@sce\}.carleton.ca}
\author{
    \IEEEauthorblockN{Mohammad M. Karbasioun\IEEEauthorrefmark{1}, Gennady Shaikhet\IEEEauthorrefmark{2}, Evangelos Kranakis\IEEEauthorrefmark{3}, Ioannis Lambadaris\IEEEauthorrefmark{1}}\\
    Carleton University, Ottawa, ON, Canada\\
	\IEEEauthorblockA{\IEEEauthorrefmark{1}Department of Systems and Computer Engineering
         \{mkarbasi, ioannis\}@sce.carleton.ca}\\
    \IEEEauthorblockA{\IEEEauthorrefmark{2}School of Mathematics and Statistics
         \{gennady\}@math.carleton.ca}\\
    \IEEEauthorblockA{\IEEEauthorrefmark{3}School of Computer Science
         \{kranakis\}@scs.carleton.ca} %\vspace{-4mm}
}			
\date{}                                           % Activate to display a given date or no date
%
%\begin{document}
\maketitle
%
%\begin{abstract}
%Motivated from the smart grid we introduce the power assignment problem for the set $\mathcal{N} $ composed of $n$ consumers, each with given energy demands $A_i$ and flexible consumption durations, $s_i$ which can be stretched and squeezed in the interval $[\ell,r]$. Our goal is to minimize the maximum power needed to serve these demands in the serving time interval $[0,1]$. Since the resulting problem is very related to the strip packing problem, we call it Power Strip Packing, where each demand is represented as a rectangle with area $A_i$ and width $s_i$.
%Our main results are as follows: There is a linear time algorithm $\mathcal{P}$ which produces a packing of $\mathcal{N} $ in a strip of width $1$ and maximum power (height) at most: $ P_{opt}(\mathcal{N}) + \dfrac{A_{max}}{\ell} $ (where $ A_{max}= \max_i \{ A_i \} $). Also by using the concept of Fractional Strip Packing, we will show that $P_{opt} \in [\dfrac{A}{Z^*}, \dfrac{A}{Z^*} + \dfrac{A_{max}}{\ell} ]  $, where $P_{opt}$ is the optimum (minimum) value for maximum power
%%, () will cannot be less than $\dfrac{A}{Z^*} $ where 
%and $A$ is the total areas of the demands and $Z^*$ is the largest achievable value obtained from combining different time durations in the interval $[0,1]$. The resulting algorithm is linear in time and asymptotically optimum.
%\end{abstract}
%
%\vspace{-2mm}
\begin{abstract}
We consider a problem of supplying electricity to a set of $\mathcal{N}$ customers in a smart-grid framework.
Each customer requires a certain amount of electrical energy which has to be supplied during the time interval $[0,1]$. We assume that each demand has 
to be supplied without interruption, with possible duration between $\ell$ and $r$, which are given system parameters ($\ell\le r$). At each moment of time,
the power of the grid is the sum of all the consumption rates for the demands being supplied at that moment. Our goal is
to find an assignment that minimizes the {\it power  peak} - maximal power over $[0,1]$ - while satisfying all the demands. To do this first we find the lower bound of optimal power  peak. We show that the problem depends on whether or not the pair $\ell, r$ belongs to a "good" region $\mathcal{G}$. If it does -  then an optimal assignment almost perfectly "fills" the rectangle $time \times power = [0,1] \times [0, A]$ with $A$ being the sum of all the energy demands - thus achieving an optimal power  peak $A$. Conversely, if $\ell, r$ do not  belong to $\mathcal{G}$, we identify the lower bound  $\overline{A} >A$ on the optimal value of power  peak
and introduce a simple linear time algorithm that almost perfectly arranges all the demands in a rectangle $[0, A /\overline{A}] \times [0, \overline{A}]$ and show that it is asymptotically optimal.
\end{abstract}
%\vspace{-4mm}
\section{Introduction}\label{SecIntroduction}
\subsection{Motivations}
One of the main goals of Demand Side Management (DSM) in smart grid is to reduce the peak to average ratio (PAR) of the total consumed power of the network \cite{GridWise}. It has several benefits for the grid such as reducing the amount of additional supplementary power needed to satisfy the demand during peak hours which itself results in decreasing the $CO_2$ emissions of power plants. Also reducing the PAR decreases the possibility of power outage due to sudden increase of demands. 
%Furthermore reducing the amount of supplementary power results in decreasing the $CO_2$ emissions of power plants. 
Furthermore with the expected presence of plug-in hybrid electric vehicles (PHEVs) to the market, finding a proper scheduling of the electric demands becomes even more crucial, since during charging hours of PHEVs the average household load is expected to be doubled \cite{GridFuture}. Also by increasing the number of PHEVs, it is vital to design and deploy suitable charging stations with proper scheduling scheme for serving the customers \cite{ChargingStationConf}. Therefore finding a proper scheduling algorithm resulting in a suitable PAR is very important in the smart grid. Different approaches and schemes are defined in this field, such as direct load control (DLC) \cite{DLC} (in which the utility can directly control the energy consumption of some loads in customer side) and real time pricing (RTP) \cite{RTP} (which propose different pricing schemes during a day). In addition in \cite{RadWung2010} authors propose a distributed algorithm using game theory. In their model the power network consists of demands each with its own energy demand and also with its own minimum and maximum acceptable power level and should be scheduled in their own requested time intervals. Even though the goal of their algorithm is reducing the total convex cost of the network, they still show that their algorithm can be used to reduce the PAR in the network.
% Unfortunately the complexity of their algorithm is relatively high, especially with large number of demands in the network. 
The problem in \cite{RadWung2010} is \emph{off-line} Scheduling of demands, since complete knowledge of the demands, such as number of demands and the amount of energy needed by each of them are known in advance. Other scenarios can happen when we don't have complete knowledge of the demands and hence we need to perform \emph{on-line} scheduling. For example in \cite{Incentive} authors try to find some algorithms for different levels of knowledge about the demands such as arrival times, durations and power intensities. Also in \cite{Tassiulas} and \cite{ElasticDemands} it is attempted to devise on-line scheduling for the random arrival of demands. In this paper we focus on off-line scheduling with a complete knowledge of the demands prior to performing scheduling. 
%\vspace{-2mm}
\subsection{Model and Related Works}
Consider a set $\mathcal{N} = \{1,2,...,n\} $ of energy demands $\{A_i, i \in \mathcal{N}\}$, needed to be scheduled in a finite time interval $[0,1]$. We assume that only "rectangular" shape of scheduling is permitted, meaning that each demand $i \in \mathcal{N}$ has to be supplied without interruption in some interval $[\tau_i, \tau_i+s_i]$ with a constant power intensity $d_i = \frac{A_i}{s_i}$.  Obviously, 
\begin{align}\label{1}
0\le \tau_i \le \tau_i+s_i\le 1, \qquad i\in \mathcal{N}. 
\end{align}
In addition to \eqref{1}, we impose a {\em demand malleability} constraint. That is, we assume the power system has parameters $\ell$ and $r$ with $0\le \ell \le r\le 1$ so that  
\begin{align}\label{2}
\ell \le s_i \le r, \qquad i\in \mathcal{N}. 
\end{align}
 This is motivated by existence electric appliances with flexibility on the charging rate. As a special application we may think of a PHEV parking lot with charging facilities where all customers either are known a priori or should register before parking their vehicles. 

A set of pairs $\pi=\{(\tau_i, s_i), i\in \mathcal{N}\}$, satisfying \eqref{1}--\eqref{2} for given parameters $(\ell, r)$, will be called a {\it scheduling policy}. Let $\Pi=\Pi^{(\ell,r)}$ be a set of all policies.

For a policy $\pi\in \Pi$, we define its maximal power as  
\begin{align}\label{eqnMaxPower}
 P_{max}^{\pi} =\max_{t\in [0,1]} \{ P^{\pi}(t)\}, 
\end{align} 
where, for  $(\tau_i,s_i)\in \pi$,
$$P^{\pi}(t)=\sum_{i=1}^{n}\left( \dfrac{A_i}{s_i} \cdot 1_{ \{ \tau_i \leq t \leq \tau_i + s_i \} } \right), \;\; 0\le t\le 1.$$ 
We are interested in finding a scheduling that minimizes the peak to average ratio in the power grid, i.e., in finding 
\begin{align}\label{3}
P_{opt} = \inf_{\pi\in\Pi}P_{max}^{\pi}.
\end{align}
Note, that $P_{opt}$, although it is not stated explicitly, depends on the parameters $\ell$ and $r$.

%\vspace{-2mm}
\subsection{Related literature}\label{SubRelatedLitrature}
Our setting resembles a so-called strip-packing problem \cite{Coeffman1980} and \cite{kenyon2000}. Indeed by viewing the demands as rectangles,
we want to pack them with their side $s_i$ parallel
to horizontal axis in a rectangular bin of width $[0,1]\times P_{opt}$
where an optimal height $P_{opt}$ is unknown. The problem is known to be 
NP complete (see \cite{lodi2002}) and therefore
an optimal height cannot always be computed
in polynomial time. 

In the traditional strip packing (TSP) problem, however, the height at any time t ($H^{\pi} (t) $) for a scheduling policy $\pi$ is defined as the {\it uppermost boundary} of scheduled rectangles at time $t$, 
%the maximal height $H^{\pi}_{max}$, for a scheduling policy $\pi$, is defined as the {\it uppermost boundary} of scheduled rectangles, 
while in our model (power strip packing - PSP) the height of the strip packing is obtained from Equation (\ref{eqnMaxPower}). This difference arises from the nature of the electric power, in which the overall height (i.e. power) at any given time is the sum of the scheduled (i.e. active) demands (See Figure \ref{fig:PSPvsTSP}). 
Naturally, for any policy $\pi$,  $P^{\pi} (t)\le H^{\pi}(t)$  at any time $t$ and hence $P^{\pi}_{max}\le H^{\pi}_{max}$.
% $P^{\pi}_{max}\le H^{\pi}_{max}$ for any policy $\pi$.

%
%
%In the next sections we will show that our algorithm and analysis work for both Traditional Strip Packing and Power Strip Packing. (In the sequel we denote traditional strip packing and power strip packing with $TSP$ and $PSP$ respectively). Let  and their corresponding optimal solutions with $H_{opt}$ and $P_{opt}$ respectively. Obviously $A \leq P_{opt} \leq H_{opt}$).
\begin{figure}[b]
\centering
\includegraphics[width=.35 \textwidth ] {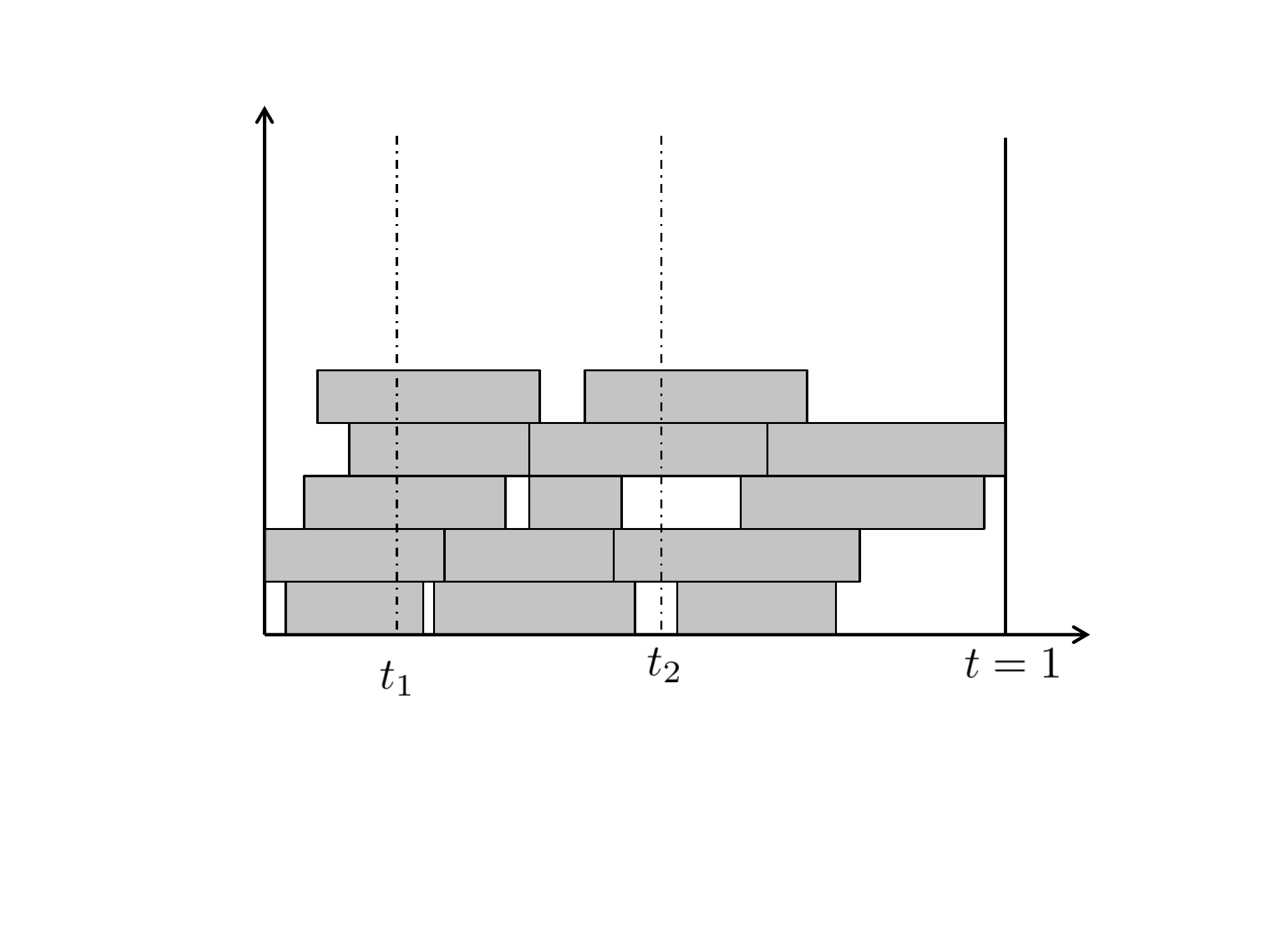} \vspace{-2mm}
\caption{Different interpretations of power and height in PSP and TSP.  For simplicity the height of each rectangle is assumed to be 1. In PSP: $P(t_1) = 5 $ and $P(t_2) = 3 $ where in TSP $H(t_1) = H(t_2) = 5 $.}\vspace{-4mm}
\label{fig:PSPvsTSP}
\end{figure}
%
%For a given set of demands, $ \mathcal{N}$, let $P_{opt}(\mathcal{N})$ and $\mathcal{P}(\mathcal{N})$  denote the optimum solutions of $PSP$ and the output of the algorithm $\mathcal{P}$ respectively. The absolute performance ratio of $\mathcal{P}$ is defined as $\sup_{\mathcal{N}} \mathcal{P}(\mathcal{N}) /  P_{opt}(\mathcal{N}) $ and the asymptotic performance ratio of $\mathcal{P}$ is defined as $\lim_{P_{opt}(\mathcal{N}) \rightarrow  \infty }  \sup_{\mathcal{N}} \mathcal{P}(\mathcal{N}) /  P_{opt}(\mathcal{N})$. (These definitions are also valid for the case of $TSP$. In that case just substitute $ PSP$ with $TSP$ and also $P_{opt}(\mathcal{N})$ with $H_{opt}(\mathcal{N})$). 

For traditional strip packing problem, most commonly referred work is \cite{kenyon2000}, which introduces an algorithm AFPTAS, whose performance $AFPTAS(\mathcal{N})$ satisfies $AFPTAS(\mathcal{N}) \leq  (1+\epsilon)H_{opt} +  O(1 /\epsilon^2)$. The algorithm is polynomial in both $n$ (the number of demands) and $(1 /\epsilon)$ where the running time is $ O(n \log n + \log^{3}{n \epsilon^{-6}} ~\log^{3}{\epsilon^{-1}} ) $. Using AFPTAS, Jansen \cite{MalleableJansen} generalized the setting of \cite{kenyon2000} to \emph{"Malleable Tasks"} where each task could use different number of resources (e.g. processor, memories,...) which can also alternate their service times. Both algorithms, especially the one in \cite{MalleableJansen}, are complex. In our case, however, a much simpler, linear-time algorithm is used to obtain a scheduling policy $\pi$ with a performance $P^{\pi}_{max}$ satisfying
\begin{align}\label{5}
P^{\pi}_{max}\le P_{opt} + \frac{A_{max}}{\ell},\quad A_{max}=\max_{i\in \mathcal{N}} A_i;
\end{align}
which, using the language of \cite{kenyon2000} and \cite{MalleableJansen}, corresponds to performance ratio being exactly $1$ (as an opposite to $1+\epsilon$).

In next section we will present Theorem \ref{ThmMain} which is the main result of the paper. Then in Sections \ref{SecAnalysis} and \ref{SecAlgorithm} we analyse the problem to find proper lower bound and upper bound for $P_{opt} $. Then based on theses results in Section \ref{SecProofThm1} we will prove Theorem \ref{ThmMain}. Finally in Section \ref{SecSimulation} we will present some computational results.
%%%%%%%%%%%%%%%%%%%%%%%%%%%%%%%%%%%%%%%%%%%%%%%%%%%%%%%%%%%%%%%%%%%%%%%%%%%%%%%%%%%%%%%%%%%%%%%%%%%%%%%%%%%%%%%%%%%%%%%%%%%%%%%%%%%%%%%%%%%%%%%%%%%%%%%%%%%%%%%%%%%%%
%%%%%%%%%%%%%%%%%%%%%%%%%%%%%%%%%%%%%%%%%%%%%%%%%%%%%%%%%%%%%%%%%%%%%%%%%%%%%%%%%%%%%%%%%%%%%%%%%%%%%%%%%%%%%%%%%%%%%%%%%%%%%%%%%%%%%%%%%%%%%%%%%%%%%%%%%%%%%%%%%%%%%
%%%%%%%%%%%%%%%%%%%%%%%%%%%%%%%%%%%%%%%%%%%%%%%%%%%%%%%%%%%%%%%%%%%%%%%%%%%%%%%%%%%%%%%%%%%%%%%%%%%%%%%%%%%%%%%%%%%%%%%%%%%%%%%%%%%%%%%%%%%%%%%%%%%%%%%%%%%%%%%%%%%%%
\vspace{-1mm}
\section{Main results}\label{SecMainResult}
We start with some preparatory work. Let $\ell$ and $r$ be fixed. A real number $w>0$ is called \textbf{achievable}, if it can be represented as an integer combination
of numbers from the interval $[\ell,r]$. That is, if there exist an integer number $q$ and a set of positive real values: $s_1,...,s_q \in [\ell, r]$ such that $ \sum_{i=1}^q s_i =w $.
 \begin{lemma}\label{LemmaWachievable}
A value $w$ is achievable if and only if  $\left\lceil \frac{w}{r} \right\rceil \leq  \frac{w}{\ell}  $.
%\begin{align}\label{eqnWachievable}
%\left\lceil \frac{w}{r} \right\rceil \leq  \frac{w}{\ell} 
%\end{align}
\end{lemma}
%
%\vspace{.5cm}
%
The proof is given in Appendix \ref{ApxLemmaWachievable}.

Lemma \ref{LemmaWachievable} implies that the point $w=1$ is achievable if $\left\lceil \frac{1}{r} \right\rceil \leq  \frac{1}{\ell}  $.
%\begin{align}
%\label{eqnCondZ^*=1}
% \left\lceil \dfrac{1}{r} \right\rceil \leq  \dfrac{1}{\ell} . 
%\end{align}
Define a {\bf good region} to be the set of all pairs $(\ell, r)$ that makes $w=1$ an achievable point (see Fig. \ref{fig:GoodRegion}):
\begin{align}
%\label{good}
\mathcal{G}=\{(\ell, r),\; \ell\le r\; \textrm{and $1$ is achievable} \}.\nonumber
\end{align}

If $w$ is not achievable, we define {\it its largest achievable point} $w^*=w^*(w)$ as:
 \begin{align}
 %\label{closest}
w^* = \sup \{v\;: v < w,\; \textrm {such that $v$ is achievable}\} \nonumber
 \end{align}

%\vspace{-4mm}
\begin{figure}[tp]
\centering
\includegraphics[width=.35 \textwidth] {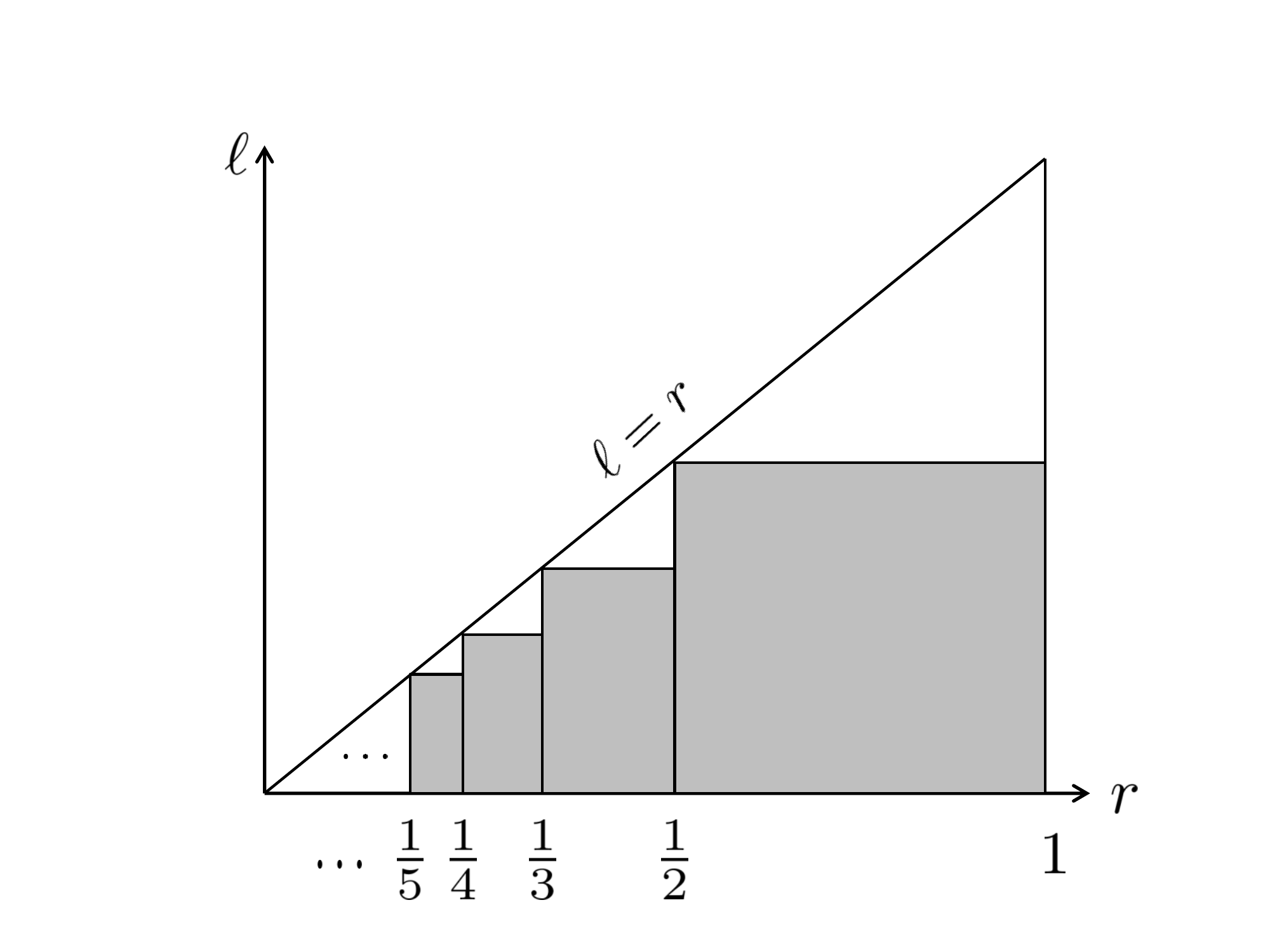}
\caption{Good region vs. Bad region. The Good region is shown with shadow.}\vspace{-4mm}
\label{fig:GoodRegion}
\end{figure}
%
%\vspace{.5cm} 

%%%%%%%%%%%%%%%%%%%%%%%%%%%%%%%%%%%%%%%%%%%%%%%%%%%%%%%%%%%%%%%%%%%%%%%%%%%%%%%%%%%%%%%%%%%%%%%%%%%%%%%%%%%%%%%%%%%%%%%%%%%%%%%%%%%%%%%%%%%%%%%%%%%%%%%%%%%%%%%%%%%%%
%%%%%%%%%%%%%%%%%%%%%%%%%%%%%%%%%%%%%%%%%%%%%%%%%%%%%%%%%%%%%%%%%%%%%%%%%%%%%%%%%%%%%%%%%%%%%%%%%%%%%%%%%%%%%%%%%%%%%%%%%%%%%%%%%%%%%%%%%%%%%%%%%%%%%%%%%%%%%%%%%%%%%
%%%%%%%%%%%%%%%%%%%%%%%%%%%%%%%%%%%%%%%%%%%%%%%%%%%%%%%%%%%%%%%%%%%%%%%%%%%%%%%%%%%%%%%%%%%%%%%%%%%%%%%%%%%%%%%%%%%%%%%%%%%%%%%%%%%%%%%%%%%%%%%%%%%%%%%%%%%%%%%%%%%%%
\begin{lemma}\label{LemmaWclosest}
 If $w$ is not achievable, then $w^*(w) = r \cdot \left\lfloor \frac{w}{r} \right\rfloor$.
\end{lemma} 
The proof is given in Appendix \ref{ApxLemmaWclosest}.

\vspace{.24cm}

 Using Lemma \ref{LemmaWclosest}, if $(\ell,r) \notin \mathcal{G}$, then $Z^*=Z^*(1) = r \cdot \lfloor \frac{1}{r} \rfloor $. 
%\begin{align}\label{eqnZvalue}
%&Z^* = r \cdot \left\lfloor \dfrac{1}{r} \right\rfloor
%\end{align}
We are ready to state the main result of the paper.
%\vspace{.5cm}
%
\vspace{.24cm}

\begin{theorem}\label{ThmMain} Let $A=\sum_{i=1}^n A_i$. Then $P_{opt}\in \left[\overline{A}, \; \overline{A}+\frac{A_{max}}{\ell}\right]$, where  
\begin{align}\label{6}
\overline{A} = \begin{cases} A &\;\;,\;\;   \textrm{if $(\ell, r)\in \mathcal{G}$},\\ \dfrac{A}{Z^*}&\;\;,\;\; \textrm{if $(\ell, r)\notin \mathcal{G}$}. \end{cases}
\end{align}
\end{theorem}
 As the number of demands $n$ grows, the value of $\overline{A}$ grows as well, thus making the interval $\left[\overline{A}, \; \overline{A}+\frac{A_{max}}{\ell}\right]$ tight. Therefore, every scheduling policy $\pi$, whose corresponding value $P_{\max}^{\pi}$ belongs to that interval, will  be considered asymptotically optimal. In the last section two of such policies will be introduced.
%if $(\ell, r)\notin \mathcal{G}$, then all demands can be satisfied in a shorter 
%time interval $[0, Z^*]$.

%\vspace{.5cm} {\bf Discussion .....}
%\begin{theorem}\label{ThmMain}
%Given set $\mathcal{N} $ composed of $n$ rectangular demands with areas $A_i  ~~i=1,2,\ldots,n$ ($ A_{max}= \max_i \{ A_i \}$) and scheduling bounds $\ell, r \leq 1$, such that 
%$~~\ell \leq s_i \leq r$, where $ s_i$ is the (time) width of the $i^{th} $ demand such that $ A_i = s_i \times \dfrac{A_i}{s_i} $. There exists a linear time algorithm $PSP$ which produces a packing of $\mathcal{N} $ in a strip of width $1$ and maximum power (height) $ PSP(\mathcal{N}) $ such that:%\footnote{To be consistent with the literature we assume that $ A_i \leq \ell , ~~i=1,2,\ldots,n $, Therefore $d_i=\frac{A_i}{s_i} $ will also satisfy $d_i \leq 1 ,  ~~i=1,2,\ldots,n$.}
%$$ 
%PSP(\mathcal{N}) \leq P_{opt}(\mathcal{N}) + \dfrac{A_{max}}{\ell}
%$$
%\end{theorem}
%%%%%%%%%%%%%%%%%%%%%%%%%%%%%%%%%%%%%%%%%%%%%%%%%%%%%%%%%%%%%%%%%%%%%%%%%%%%%%%%%%%%%%%%%%%%%%%%%%%%%%%%%%%%%%%%%%%%%%%%%%%%%%%%%%%%%%%%%%%%%%%%%%%%%%%%%%%%%%%%%%%%%
%%%%%%%%%%%%%%%%%%%%%%%%%%%%%%%%%%%%%%%%%%%%%%%%%%%%%%%%%%%%%%%%%%%%%%%%%%%%%%%%%%%%%%%%%%%%%%%%%%%%%%%%%%%%%%%%%%%%%%%%%%%%%%%%%%%%%%%%%%%%%%%%%%%%%%%%%%%%%%%%%%%%%
%%%%%%%%%%%%%%%%%%%%%%%%%%%%%%%%%%%%%%%%%%%%%%%%%%%%%%%%%%%%%%%%%%%%%%%%%%%%%%%%%%%%%%%%%%%%%%%%%%%%%%%%%%%%%%%%%%%%%%%%%%%%%%%%%%%%%%%%%%%%%%%%%%%%%%%%%%%%%%%%%%%%%
\vspace{-1mm}
\section{Lower bound of $P_{opt} $} \label{SecAnalysis}
Theorem \ref{ThmMain} will be proved in two main steps, each dealing with lower and upper bounds respectively. Here we discuss about lower bound.

\begin{lemma}\label{LemAZ}
\begin{align}\label{IneqClaim}
 P_{opt} \ge \overline{A}.
\end{align}
\end{lemma}
\proof
The Inequality \eqref{IneqClaim} is obvious for $(\ell, r)\in \mathcal{G}$ (making $Z^* =1$ and $\overline{A}=A$).
Therefore, in what follows we assume $(\ell, r)\notin \mathcal{G}$ and hence $Z^* <1$. To prove Lemma \ref{LemAZ} we use the concept of the \emph{Fractional Strip Packing (FSP)} \cite{FSP}, according to which it is allowed to perform horizontal cutting to each rectangle $(s_i,d_i)$ in order to get a set of equal width rectangles $\{ (s_i,d_{i1}),(s_i,d_{i2}) , \ldots , (s_i,d_{iQ_{i}}) \} $ such that $ \sum_{j=1}^{Q_i} d_{ij} =d_i$ for all $i=1,\dots ,n $. Then, instead of packing the original set $\{ (s_i,d_i) ~~~i=1,\ldots,n \}$ of rectangles, a packing of a newly obtained set is performed. 

Assume a certain set $\mathcal{N}$ of demands is given. For any FSP policy $\theta$, (that includes both assigning $(s_i, d_i)$, $i\in \cal{N}$, as well as cutting), let  $P_{max}^{\theta}$ and $H_{max}^{\theta}$ denote the maximal achieved heights,
calculated according to traditional and power strip packing respectively. And let also $P_{opt}^{F}$ and $H_{opt}^{F}$ denote the FSP optimal values for traditional and power strip packing. Clearly, $ A \leq P_{opt}^{F}\le H_{opt}^{F}$ and also 

\begin{align} \label{FSP_P}
 &A \leq P_{opt}^{F} \leq P_{opt}\;,  \\
\label{FSP_H}
&A \leq H_{opt}^{F} \leq H_{opt} \;. 
\end{align}
%In the sequel we try to show that $\dfrac{A}{Z^*} \leq \leq FSP_{P,opt} (\mathcal{N}) \leq FSP_{H,opt} (\mathcal{N})$.   
Next we will show that $ \frac{A}{Z*} \leq P_{opt}^{F}$, which, combined with \eqref{FSP_P}, proves Lemma \ref{LemAZ}.

Assume an arbitrary scheduling policy $\pi\in \Pi^{(\ell, r)}$ is given.  Cut each demand $ (s_i,d_i) $ horizontally to get the new set of very short height rectangles
$\mathcal{N}' = \{ (s_i,d_{i1}),(s_i,d_{i2}) , \ldots , (s_i,d_{iQ_{i}}) \} $ such that $d_{ij} = \delta  ~~\forall i,j $ where $\delta$ is an infinitesimal  positive real value. Now we try to pack these narrow rectangles in a strip of width $1$ as follows:
Let call an FSP policy $\theta^*$ as {\emph{filling} when there will not remain space to put any new rectangle with any arbitrary width $ s_i$ ($\ell \leq s_i \leq r $) in each of its rows, i.e. the sum of all gaps in each row is always less than $\ell$. (maybe except for the last row). ($\Theta^* $ is the set of all possible fillings $\theta^* $). Now to pack all the narrow rectangle with a \textit{filling}, we start with the first narrow rectangle and put it in the first row and then try to \emph{fill} this row with subsequent narrow rectangles. After filling a row we continue with filling subsequent rows until all the narrow rectangles will be scheduled.

\vspace{2mm}
\begin{lemma}\label{LemmaFillingOpt}
\begin{align} \label{eqnLemma4}
P_{opt}^{F} = \min _{\theta^* \in \Theta^*} {P_{max}^{\theta^*}}
\end{align}
\end{lemma}
\vspace{2mm}

The proof is given in Appendix \ref{ApxLemmaFillingOpt}. So with respect to Equation \eqref{eqnLemma4}, in the sequel to show that $ \frac{A}{Z*} \leq P_{opt}^{F}$ we just consider fillings $\theta^* \in \Theta^*$.

From Lemma \ref{LemmaWclosest} we know that in each row we cannot cover the time axis more than $Z^* $, then the maximum height (i.e. the upper boundary of the top row) cannot be less than $\frac{A}{Z^*}$, otherwise the total area will be less than $A$. Combining this fact with Inequality (\ref{FSP_H}) results in:
\begin{align}\label{T1212}
   \frac{A}{Z^*} \leq H_{opt}^{F} \leq H_{opt}
\end{align} 
It proves Lemma \ref{LemAZ} for TSP, but to prove this Lemma for PSP we should take into account all possible \emph{gaps} in each row. Because, as mentioned before, the way of computing the maximum height in PSP is different from TSP. Indeed in PSP in calculating the height or more precisely the total power at each time, only the active demands at that time will be considered and \emph{gaps} are not taken into account (Equation \eqref{eqnMaxPower} and Figure \ref{fig:PSPvsTSP}). Also note that by \emph{filling} each row arbitrarily, these gaps can split in more than just one part in a row and may happen in different places of the time axis for different rows. 
Lemma \ref{LemmaKrectangles} tells that in any arbitrary \emph{filling} there will be exactly $K_0 = \lfloor \frac{1}{r} \rfloor$ rectangles in each row and then in Lemma \ref{LemmaKnoGaps} it will be shown that for any arbitrary \emph{filling}, there will be exactly $K_0$ equal length and identically placed \emph{``gap-free''} interval in each row. 
%It means that there are $ K_0$ intervals in time axis for which there are scheduled narrow rectangles in each row. 
Therefore for each filling $ \theta^*$, the resulted $P_{max}$ here is indeed equal to $H_{max}$ and consequently $ P_{opt}^{F} = H_{opt}^{F}$. Eventually using Lemma \ref{LemmaFillingOpt} together with \eqref{T1212} proves Lemma \ref{LemAZ} for PSP.
%Therefore the resulted $P_{max}$ here is indeed equal to $H_{max}$. Finally by using Lemma \ref{LemmaFillingOpt} we have: $ P_{opt}^{F} = H_{opt}^{F}$, and hence Lemma \ref{LemAZ} is also valid for PSP.

\begin{lemma}\label{LemmaKrectangles}  
In any filling $\theta^*$ for a set of narrow rectangles of arbitrary widths $s_i$ (s.t. $ \ell \leq s_i \leq r $), exactly $K_0 = \lfloor \frac{1}{r} \rfloor = \lfloor \frac{1}{\ell} \rfloor $ narrow rectangles are needed (This may not apply for the last row which may contain less rectangles).
\end{lemma}
The proof is given in Appendix \ref{ApxLemmaKrectangles}.

\begin{lemma}\label{LemmaKnoGaps}
Under the same conditions as in Lemma \ref{LemmaKrectangles}, each row will contain $K_0$ identical (i.e. equal length and identically placed) intervals, $ (1-(K_0 - i +1) \cdot \ell ~, ~i \cdot \ell ) $ for all $ i=1,\ldots, K_0 $, which cannot contain any gaps (maybe except for the last row).
\end{lemma}

\begin{proof}
Assume that we are trying to put $K_0$ arbitrary width narrow rectangles in a row. Now assume that we want to put the $i^{th}$ rectangle in this row. The largest value for the starting time of $i^{th}$ narrow rectangle, i.e. $\tau_i$, occurs when all the narrow rectangles starting after $\tau_i$, including $i^{th}$ rectangle itself, have the same width $\ell$. Also all possible gaps happen before $\tau_i$. Then the length of the remaining time after $\tau_i$ is : $( K_0 -(i-1) ) \cdot \ell $ and hence the largest value for $\tau_i$ is as follows: 
$$\tau_{i,largest}=1 - ( K_0 -(i-1) ) \cdot \ell $$
On the other hand the smallest value for the finishing time of this rectangle, i.e. $f_i=\tau_i + s_i $,  happens only when this rectangle and all of the previous $(i-1)$ narrow rectangles have the same width $\ell$ and no gap placed in the time axis until after the $f_i $. Then the smallest value for $f_i$ is: $$f_{i,smallest}= i \cdot \ell $$
Therefore $i^{th}$ narrow rectangle starts at most at time  $\tau_{i,largest} $ and remains active at least until $f_{i,smallest} $. Then we conclude that the $i^{th}$ narrow rectangle (and \textbf{only} this one) is definitely active during the following time interval $ ( 1-(K_0 - i +1) \cdot \ell , i \cdot \ell ) $ with the length:
\begin{align}
&i \cdot \ell-  (1-(K_0 - i +1) \cdot \ell)   \nonumber \\
= &(K_0 +1) \cdot \ell -1 =  \left \lceil \frac{1}{\ell}  \right \rceil \cdot \ell -1 > 0  \nonumber
\end{align}

Therefore we have $K_0 = \lfloor \frac{1}{r} \rfloor = \lfloor \frac{1}{\ell} \rfloor $ active intervals with equal non-zero length $ (K_0 +1) \cdot \ell -1  $.
\end{proof}

Note that it is impossible for a demand to be completely placed in the interval between two consecutive active intervals, because the length of such an interval is $1- K_0 \cdot \ell $ which is less than $\ell$ (because $ 1 < \lceil \frac{1}{\ell} \rceil = (K_0 + 1) \cdot \ell$).

\begin{corollary}\label{CollaryInterval}
In $i^{th}$ active interval of any filling $\theta^*$ (Lemma \ref{LemmaKnoGaps}), when filling each row with narrow rectangles, definitely the $i^{th}$ narrow rectangle in that row is scheduled and the $i^{th}$ narrow rectangle is the only one which can be scheduled in this interval (maybe except for the last row).
\end{corollary}
Using Lemma \ref{LemmaKnoGaps}, we conclude that in FSP when \emph{filling} the rows with any arbitrary placement and with any arbitrary width ($s_i$) for narrow rectangles, there will be exactly $K_0$ gap-free \emph{identical} intervals in every row. Therefore the $P_{max}^{\theta^*}$ is exactly the same as $H_{max}^{\theta^*}$ which is the sum of the heights of all the rows. 

Note that no one can change FSP above and get the lower $P_{max}^{\theta^*}$. For example if one tries to decrease the $P_{max}^{\theta^*}$ by just one level, it needs to remove at least $K_0$ narrow rectangles from different gap-free intervals, either from one  row or $K_0$ different rows (corollary \ref{CollaryInterval}). Then these narrow rectangles will form at least one new row which then only can keep the total height the same as before and cannot decrease the $P_{max}^{\theta^*}$. So $P_{max}^{\theta^*}$ above is in fact minimum among all possible \emph{fillings} and hence using \eqref{eqnLemma4} in Lemma \ref{LemmaFillingOpt} it is equal to $P_{opt}^{F}$.

From Lemma \ref{LemmaWclosest} we know that in a filling of narrow rectangles, in each row we cannot cover the time axis more than $Z^* = r \cdot \lfloor \frac{1}{r} \rfloor $, then the maximum height (i.e. the upper boundary of the top row) cannot be less than $\frac{A}{Z^*}$, otherwise the total area will be less than $A$. Therefore we have: $  \frac{A}{Z^*} \leq P_{opt}^{F} \leq P_{opt}$ which proves Lemma \ref{LemAZ}.  \qed
 
In the next section we will use a simple linear time algorithm, in which we try to keep the $P(t)$ (and hence $P_{max}$) around $\frac{A}{Z^*}$. Using this algorithm we will show that the $ P_{opt} \le \overline{A}+\frac{A_{max}}{\ell} $ which proves the upper bound of Theorem \ref{ThmMain} and hence it also proves that this algorithm is asymptotically optimum.
%%%%%%%%%%%%%%%%%%%%%%%%%%%%%%%%%%%%%%%%%%%%%%%%%%%%%%%%%%%%%%%%%%%%%%%%%%%%%%%%%%%%%%%%%%%%%%%%%%%%%%%%%%%%%%%%%%%%%%%%%%%%%%%%%%%%%%%%%%%%%%%%%%%%%%%%%%%%%%%%%%%%%
%%%%%%%%%%%%%%%%%%%%%%%%%%%%%%%%%%%%%%%%%%%%%%%%%%%%%%%%%%%%%%%%%%%%%%%%%%%%%%%%%%%%%%%%%%%%%%%%%%%%%%%%%%%%%%%%%%%%%%%%%%%%%%%%%%%%%%%%%%%%%%%%%%%%%%%%%%%%%%%%%%%%%
\vspace{-2mm}
\section{Algorithm}\label{SecAlgorithm}
From Lemma \ref{LemmaWachievable}, we know that $t=1$ is achievable if and only if $ \lceil \frac{1}{r} \rceil \leq  \frac{1}{\ell}  $. Based on this fact, in the sequel, we will introduce different cases and in each case we will find a proper policy which results in acceptable $P_{max} $. 
%%%%%%%%%%%%%%%%%%%%%%%%%%%%%%%%%%%%%%%%%%%%%%%%%%%%%%%%%%%%%%%%%%%%%%%%%%%%%%%%%%%%%%%%%%%%%%%%%%%%%%%%%%%%%%%%%%%%%%%%%%%%%%%%%%%%%%%%%%%%%%%%%%%%%%%%%%%%%%%%%%%%%
%%%%%%%%%%%%%%%%%%%%%%%%%%%%%%%%%%%%%%%%%%%%%%%%%%%%%%%%%%%%%%%%%%%%%%%%%%%%%%%%%%%%%%%%%%%%%%%%%%%%%%%%%%%%%%%%%%%%%%%%%%%%%%%%%%%%%%%%%%%%%%%%%%%%%%%%%%%%%%%%%%%%%
%%%%%%%%%%%%%%%%%%%%%%%%%%%%%%%%%%%%%%%%%%%%%%%%%%%%%%%%%%%%%%%%%%%%%%%%%%%%%%%%%%%%%%%%%%%%%%%%%%%%%%%%%%%%%%%%%%%%%%%%%%%%%%%%%%%%%%%%%%%%%%%%%%%%%%%%%%%%%%%%%%%%%
%%%%%%%%%%%%%%%%%%%%%%%%%%%%%%%%%%%%%%%%%%%%%%%%%%%%%%%%%%%%%%%%%%%%%%%%%%%%%%%%%%%%%%%%%%%%%%%%%%%%%%%%%%%%%%%%%%%%%%%%%%%%%%%%%%%%%%%%%%%%%%%%%%%%%%%%%%%%%%%%%%%%%
\vspace{-2mm}
\subsection{Ideal Cases}\label{SubSection:Ideal}
First assume that $ \ell \leq 1 \hspace{0.5cm} \mbox{and} \hspace{0.5cm} r = 1 $ (So $1$ is achievable). In this case we have an optimal policy $ \pi^* $ as follows: stretch the width of each demand to the whole time interval $[0,1]$, and then piling the demands up on top of each other. Another ideal case is when we have $ \ell \leq \frac{A_{i}}{A} \leq r$,  $ i=1,\cdots,n $. In this case in optimal policy $ \pi^*$ we have: $s_i= \frac{A_i}{A} $, therefore since $ \sum_{i=1}^n s_i = \sum_{i=1}^n \frac{A_i}{A} = 1 $, $1$ is achievable. So in this case we can simply stretch each intensity to $A$ as $A_i = \frac{A_i}{A} \times A$ and place the demands side by side. In both cases we have $ P(t)=A~~~ \mbox{for all $t \in [0,1] $} $ and hence:
\begin{align}\label{7}
P_{\max}^{\pi^*}= A \le \overline{A}+\dfrac{A_{max}}{\ell} 
\end{align}

So for Ideal cases, we can pack the demands in a strip of height $A$ and  and because $A \leq P_{opt} \leq H_{opt}$, the output in this case is absolutely optimum. 
% so the Theorem \ref{ThmMain} is proved for this case.
% 
%%%%%%%%%%%%%%%%%%%%%%%%%%%%%%%%%%%%%%%%%%%%%%%%%%%%%%%%%%%%%%%%%%%%%%%%%%%%%%%%%%%%%%%%%%%%%%%%%%%%%%%%%%%%%%%%%%%%%%%%%%%%%%%%%%%%%%%%%%%%%%%%%%%%%%%%%%%%%%%%%%%%%
\vspace{-2mm}
\subsection{Near Ideal Cases}\label{SubSection:Case2}
In this case $t=1$ is achievable and hence $Z^*=1$. Then we perform scheduling policy $\pi $ as follows: Simply divide the time interval $[0,1]$ into $K_0 = \lceil \frac{1}{r} \rceil$ non-overlapping time slots with equal length $S_0=\frac{1}{K_0}$ and since in this case we have $\lceil \frac{1}{r} \rceil \leq  \frac{1}{\ell} $ ( Lemma \ref{LemmaWachievable}), we are sure that $ \ell \leq S_0 \leq r $. 

Now we can start with the first demand and put it in the first slot and continue this until the height of that slot becomes greater or equal to the threshold value, $A$, then continue by putting the next demand in the next slot and doing the same procedure until all of the demands have been packed. Therefore the height of each slot is at most $A+\frac{A_{max}}{S_0} $. So in this case we have:
\begin{align}\label{PmaxNearIdeal}  
P_{\max}^{\pi} \leq A+\dfrac{A_{max}}{S_0} \leq \overline{A}+\dfrac{A_{max}}{\ell} 
\end{align}
%%%%%%%%%%%%%%%%%%%%%%%%%%%%%%%%%%%%%%%%%%%%%%%%%%%%%%%%%%%%%%%%%%%%%%%%%%%%%%%%%%%%%%%%%%%%%%%%%%%%%%%%%%%%%%%%%%%%%%%%%%%%%%%%%%%%%%%%%%%%%%%%%%%%%%%%%%%%%%%%%%%%%
%%%%%%%%%%%%%%%%%%%%%%%%%%%%%%%%%%%%%%%%%%%%%%%%%%%%%%%%%%%%%%%%%%%%%%%%%%%%%%%%%%%%%%%%%%%%%%%%%%%%%%%%%%%%%%%%%%%%%%%%%%%%%%%%%%%%%%%%%%%%%%%%%%%%%%%%%%%%%%%%%%%%%
%%%%%%%%%%%%%%%%%%%%%%%%%%%%%%%%%%%%%%%%%%%%%%%%%%%%%%%%%%%%%%%%%%%%%%%%%%%%%%%%%%%%%%%%%%%%%%%%%%%%%%%%%%%%%%%%%%%%%%%%%%%%%%%%%%%%%%%%%%%%%%%%%%%%%%%%%%%%%%%%%%%%%
\vspace{-8mm}
\subsection{Non-Ideal Cases}\label{SubSection:Case3}
In this case $t=1$ is NOT achievable, so the constraints on $\ell$ and $r$ are as follows:

\begin{align}
0 < \ell<r<1 \hspace{1.5cm} \text{and} \hspace{1.5cm}  
\left \lceil \frac{1}{r} \right \rceil >  \frac{1}{\ell} \nonumber 
\end{align}
These conditions imply that $ \lfloor \frac{1}{r} \rfloor = \lfloor \frac{1}{\ell} \rfloor$. Using Lemma \ref{LemmaWclosest}, we have: $Z^* = r \cdot \lfloor \frac{1}{r} \rfloor <1$.

In this case the chosen policy $\pi $ is as follows: We simply divide the time interval $[0,Z^*]$ into $K_0 = \lfloor \frac{1}{r} \rfloor $ non-overlapping time slots with equal length $r$ and leave the remaining part of the time interval, i.e. $[Z^*,1] $ \emph{unscheduled}. Now we can use the same approach used in Subsection \ref{SubSection:Case2}, while here $s_i=S_0=r,~~i=1,\ldots ,n $ and $K_0 = \lfloor \frac{1}{r} \rfloor $ and also the threshold value is $\frac{A}{Z^*}$. Then in this case we have:
\begin{align}\label{PmaxNonIdeal}  
P_{\max}^{\pi} \leq \dfrac{A}{Z^*} +\dfrac{A_{max}}{r} \leq \overline{A} + \dfrac{A_{max}}{\ell} 
\end{align}
%
%%%%%%%%%%%%%%%%%%%%%%%%%%%%%%%%%%%%%%%%%%%%%%%%%%%%%%%%%%%%%%%%%%%%%%%%%%%%%%%%%%%%%%%%%%%%%%%%%%%%%%%%%%%%%%%%%%%%%%%%%%%%%%%%%%%%%%%%%%%%%%%%%%%%%%%%%%%%%%%%%%%%%
\section{Proof of Theorem \ref{ThmMain} }\label{SecProofThm1}
\begin{proof}
Lemma \ref{LemAZ} showed that $\overline{A}\le P_{opt} $, which proves the lower bound part of Theorem \ref{ThmMain}. In addition in Section \ref{SecAlgorithm}, for each set of demands $\mathcal{N}$ with pair $(\ell,r)$, resulting in ideal, near-ideal, non-ideal cases, we introduced a policy $ \pi$ with performance $P_{\max}^{\pi} \le  \overline{A} + \dfrac{A_{max}}{\ell} $ which results in:  
%%%
%%%
%%%Also By combing the results we obtained for each of ideal, near-ideal, non-ideal cases, (i.e. inequalities \eqref{7}, \eqref{PmaxNearIdeal} and \eqref{PmaxNonIdeal} respectively) for every set of demands $\mathcal{N}$ and for all pair $(\ell,r)$, we introduced a policy $ \pi$ with performance $P_{\max}^{\pi} \le  \overline{A} + \dfrac{A_{max}}{\ell} $ which results in:  
\begin{align}\label{8}
P_{opt}\le \overline{A} + \dfrac{A_{max}}{\ell}
\end{align} 
which proves the upper bound part of Theorem \ref{ThmMain}. Therefore combining inequalities \eqref{IneqClaim} and \eqref{8} gives us: $P_{opt}\in \left[\overline{A}, \; \overline{A}+\frac{A_{max}}{\ell}\right]$ which proves Theorem \ref{ThmMain}.
\end{proof}
%%%%%%%%%%%%%%%%%%%%%%%%%%%%%%%%%%%%%%%%%%%%%%%%%%%%%%%%%%%%%%%%%%%%%%%%%%%%%%%%%%%%%%%%%%%%%%%%%%%%%%%%%%%%%%%%%%%%%%%%%%%%%%%%%%%%%%%%%%%%%%%%%%%%%%%%%%%%%%%%%%%%%
%\vspace{-2mm}
\section{Computational Results}\label{SecSimulation}
Algorithm \ref{algRectangular}, which is called ``\emph{Power Strip Packing}'' algorithm, summarizes all the steps of packing the demands. Clearly the running time of this algorithm is linear in $n$, the number of demands. As we discussed in Section \ref{SecAlgorithm}, we showed for this algorithm we have:
\begin{align}\label{9}
P_{\max}^{PSP} \le \overline{A}+\frac{A_{\max}}{\ell} \le P_{opt} + \frac{A_{\max}}{\ell}
\end{align}
In obtaining the last inequality in \eqref{9} we used Inequality \eqref{IneqClaim} in Lemma \ref{LemAZ}. So as it was mentioned in Subsection \ref{SubRelatedLitrature}, using the language of \cite{kenyon2000} and \cite{MalleableJansen}, it corresponds to performance ratio being exactly $1$. Therefore this algorithm is asymptotically optimal for all different cases. Furthermore as we can see in Algorithm \ref{algRectangular}, this algorithm is linear time.
%Also if further $ A_{max} \ll \ell$, e.g. $ A_i = O(\frac{1}{n}) $ and $n$ is very large, then $ \frac{A_{max}}{\ell} \ll 1 $.
%

We can achieve better performance in terms of flatter $ P(t)$ for $ t \in [0,1]$ and get lower $P_{max} $, albeit at the expense of increasing the number of operations. For example the greedy algorithm works by ordering the demands by non-increasing areas $(A_i)$ and select $s_i=S_0$ for $ i=1,\cdots, n$ exactly the same as that of in PSP Algorithm \ref{algRectangular} (and hence $ d_i= \frac{A_i}{S_0}$). Then start with the first (i.e. largest) demand and put it in a slot with minimum total height and continue this for subsequent demands until packing all of the demands. Note that when greedy algorithm place a new demand in a slot, it is not possible that the height of every slot being greater than $\frac{A}{Z^*} $, otherwise the total area $\frac{A}{Z^*} \times Z^* $ exceeds $A$, the sum of the demands. Therefore the height of each slot is at most $ \frac{A}{Z^*} +\frac{A_{max}}{S_0} $. So the inequalities (\ref{PmaxNearIdeal}) and (\ref{PmaxNonIdeal}) are still valid and hence it is also asymptotically optimal. However the running time is $O(n^2)$. 

Figures \ref{fig:AlgNearIdeal} and \ref{fig:AlgNonIdeal} illustrate the outputs of the presented algorithms (Algorithm \ref{algRectangular} and the greedy algorithm) for different values of $ |\mathcal{N}|$, the number of demands, in near ideal cases and non-ideal cases respectively. In these figures the demands, i.e. $A_i $'s, are independent and identically uniformly distributed in the interval $ [0,\ell] $. For each value of $ |\mathcal{N}|$, each of the algorithms is performed 30 times and then the mean  value is depicted as the corresponding $P_{\max} $ of each algorithm. Furthermore for some values of $ |\mathcal{N}|$ the corresponding $0.95\%$ confidence intervals are shown in these figures.

As it can be seen from these figures, in both cases, with or without ordering always we have: $P_{max} \leq \frac{A}{Z^*} + \frac{A_{max}}{\ell}  $. However as we can see in these figures, ordering can improve the performance of the algorithm albeit at the expense of increasing the number of operations.
%%%%%%%%%%%%%%%%%%%%%%%%%%%%%%%%%%%%%%%%%%%%%%%%%%%%%%%%%%%%%%%%%%%%%%%%%%%%%%%%%%%%%%%%%%%%%%%%%%%%%%%%%%%%%%%%%%%%%%%%%%%%%%%%%%%%%%%%%%%%%%%%%%%%%%%%%%%%%%%%%%%%%
%%%%%%%%%%%%%%%%%%%%%%%%%%%%%%%%%%%%%%%%%%%%%%%%%%%%%%%%%%%%%%%%%%%%%%%%%%%%%%%%%%%%%%%%%%%%%%%%%%%%%%%%%%%%%%%%%%%%%%%%%%%%%%%%%%%%%%%%%%%%%%%%%%%%%%%%%%%%%%%%%%%%%
%%%%%%%%%%%%%%%%%%%%%%%%%%%%%%%%%%%%%%%%%%%%%%%%%%%%%%%%%%%%%%%%%%%%%%%%%%%%%%%%%%%%%%%%%%%%%%%%%%%%%%%%%%%%%%%%%%%%%%%%%%%%%%%%%%%%%%%%%%%%%%%%%%%%%%%%%%%%%%%%%%%%%
%
%
%
\algsetup{
linenosize=\small,
linenodelimiter=.,
indent=1em
}
\begin{algorithm}
\caption{Power Strip Packing algorithm}
\label{algRectangular}
\begin{algorithmic}[1]
\STATE INPUT $\ell$,$r$, Demands $A_i$  $i=1,\ldots , n$
\STATE OUTPUT $P_{max} $,$H_{max} $ and Scheduling policy 
\STATE $A=\sum_{i=1}^n A_i$, $A_{min} = \min_i \{A_i \} $, $A_{max} = \max_i \{A_i \} $
\IF{$ (r \geq 1)$ or $ ( \ell \leq \frac{A_{min}}{A} \leq \frac{A_{max}}{A} \leq r)$}
\STATE Perform Ideal Scheduling (Subsection \ref{SubSection:Ideal}) 
\STATE $P_{max} = H_{max}=A $
\ELSE
%\STATE order the demands by non-increasing areas $A_i$'s
\STATE if $\{ \lceil \frac{1}{r} \rceil \leq  \frac{1}{\ell} \}$ then  $K_0 = \lceil \frac{1}{r} \rceil$ and $S_0=\frac{1}{K_0}$ and $ Z^* = 1 $
\STATE else $K_0 = \lfloor \frac{1}{r} \rfloor$ and $S_0=r$ and $Z^* = r \cdot \lfloor \frac{1}{r} \rfloor$

\STATE $Threshold = \frac{A}{Z^*}$ 
\STATE Fill each slot until its height exceeds the $ Threshold$, then continue with next slot
%%\IF{$ \left\lceil \dfrac{1}{r} \right\rceil \leq  \dfrac{1}{\ell}$} 
%%\STATE $K_0 = \left\lceil \dfrac{1}{r} \right\rceil$ and $S_0=\dfrac{1}{K_0}$
%%\ELSE
%%\STATE $K_0 = \left\lfloor \dfrac{1}{r} \right\rfloor$ and $S_0=r$
%%\ENDIF
%\STATE Put each demand in a slot with minimum height
\STATE $P_{max} = H_{max}= \max_{j} \{\text{height of $j $-th slot} ~~~~,j=1,\ldots,K_0 \} $
\ENDIF 
\end{algorithmic}
\end{algorithm}
%
%
%
%
%
%
%
%%%%%%%%%%%%%%%%%%%%%%%%%%%%%%%%%%%%%%%%%%%%%%%%%%%%%%%%%%%%%%%%%%%%%%%%%%%%%%%%%%%%%%%%%%%%%%%%%%%%%%%%%%%%%%%%%%%%%%%%%%%%%%%%%%%%%%%%%%%%%%%%%%%%%%%%%%%%%%%%%%
%%%%%%%%%%%%%%%%%%%%%%%%%%%%%%%%%%%%%%%%%%%%%%%%%%%%%%%%%%%%%%%%%%%%%%%%%%%%%%%%%%%%%%%%%%%%%%%%%%%%%%%%%%%%%%%%%%%%%%%%%%%%%%%%%%%%%%%%%%%%%%%%%%%%%%%%%%%%%%%%%%
%Also in figure \ref{fig:ConfInterval}, for different values of $ |\mathcal{N}|$, i.e. the number of demands, the algorithm \ref{algRectangular} is performed 30 times and then for each value of $ |\mathcal{N}|$ the mean value, Mean, is shown within corresponding $0.95\%$ confidence interval, i.e. $[Low,High] $ in this figure. 
%
%\begin{figure}[tp]
%\centering
%\includegraphics[width=0.50 \textwidth , height=2.25in] {ConIntervalNonSorted}
%\caption{Confidence interval of Threshold-based Algorithm (Number of Iterations = 30).}
%\label{fig:ConfInterval}
%\end{figure}
%%%%%%%%%%%%%%%%%%%%%%%%%%%%%%%%%%%%%%%%%%%%%%%%%%%%%%%%%%%%%%%%%%%%%%%%%%%%%%%%%%%%%%%%%%%%%%%%%%%%%%%%%%%%%%%%%%%%%%%%%%%%%%%%%%%%%%%%%%%%%%%%%%%%%%%%%%%%%%%%%%
%%%%%%%%%%%%%%%%%%%%%%%%%%%%%%%%%%%%%%%%%%%%%%%%%%%%%%%%%%%%%%%%%%%%%%%%%%%%%%%%%%%%%%%%%%%%%%%%%%%%%%%%%%%%%%%%%%%%%%%%%%%%%%%%%%%%%%%%%%%%%%%%%%%%%%%%%%%%%%%%%%
%\vspace{-4mm}
\section{Conclusion}
%Power Strip Packing (PSP) problem which is introduced in this paper Motivated from the smart grid we introduce the power assignment problem
%We consider a problem of supplying electricity to a set of $\mathcal{N}$ customers in a smart-grid framework.
In this paper considering a problem of supplying electricity to malleable demands, we introduced off-line Power Strip Packing (PSP) problem and showed that using a linear time algorithm will result in asymptotically optimal performance. One may extend this problem such that different demands can have different scheduling bounds. Another (and may be more important) extension to this problem is On-line Power Strip Packing where the goal is serving the arriving demands with different stochastic characteristics such as energy demands, scheduling bounds $(\ell, r)$, arriving times and the deadline for scheduling each of them.

%
%\begin{figure}[tp]
%\centering
%%\includegraphics[width=.50 \textwidth ] {AlgorithmsPerformancesNearIdeal}
%\includegraphics[width=.50 \textwidth ] {2}\vspace{-1mm}
%\caption{Comparing performances of the algorithms for Non-Ideal cases with  $ \frac{A}{Z^*} + \frac{A_{max}}{\ell} $ for different number of demands, where $\ell=0.35714 $ and $ r=0.75758  $ and $A_i$'s are  independent and identically uniformly distributed in the interval $ [0,\ell] $.}\vspace{-3mm}
%\label{fig:AlgNearIdeal}
%\end{figure}
%%%%%%%%%%%%%%%%%%%%%%%%%%%%%%%%%%%%%%%%%%%%%%%%%%%%%%%%%%%%%%%%%%%%%%%%%%%%%%%%%%%%%%%%%%%%%%%%%%%%%%%%%%%%%%%%%%%%%%%%%%%%%%%%%%%%%%%%%%%%%%%%%%%%%%%%%%%%%%%%%%%%%
%\appendix
%\renewcommand\thesection{Appendix \Alph{section}:}
%%%%%%%%%%%%%%%%%%%%%%%%%%%%%%%%%%%%%%%%%%%%%%%%%%%%%%%%%%%%%%%%%%%%%%%%%%%%%%%%%%%%%%%%%%%%%%%%%%%%%%%%%%%%%%%%%%%%%%%%%%%%%%%%%%%%%%%%%%%%%%%%%%%%%%%%%%%%%%%%%%%%%
\appendices
\vspace{1mm}
\section{Proof of Lemma \ref{LemmaWachievable}}\label{ApxLemmaWachievable}
\begin{proof}
Suppose that there is a sequence $s_j$, $j=1,..., q $ satisfying $\ell \leq  s_j \leq r $. Now Consider the following definition: $ s' \triangleq \frac{1}{q} \sum_{j=1}^q s_j = \sum_{j=1}^q \frac{1}{q} \cdot s_j $
%\begin{align}\label{eqnLemmaWCovexSum}
%  s' \triangleq \frac{1}{q} \sum_{j=1}^q s_j = \sum_{j=1}^q \frac{1}{q} \cdot s_j
%\end{align}
%
(Clearly $ \ell \leq  s' \leq r$).
%In the sequel instead of using $\sum_{j=1}^q s_j$, we will simply use its equivalent value: $ q \cdot s' $.
%Note that $ \sum_{j=1}^q \frac{1}{q} =1 $, so the summation in Equation (\ref{eqnLemmaWCovexSum}) is a convex combination of $s_j$'s. Therefore $s'$ also satisfies the condition $ \ell \leq  s' \leq h$. So in continue instead of using $\sum_{j=1}^q s_j$ we can simply use its equivalent value: $ q \cdot s' $.\\
Then $w$ is achievable if and only if there exist an integer value $q$ and a real value $s'$ such that $ \ell \leq s' \leq r$ and $w=q.s'$. So we have:
\begin{align}
s' = \dfrac{w}{q}  &\Longleftrightarrow  \ell \leq \dfrac{w}{q} \leq r \nonumber \\
 \Longleftrightarrow  \dfrac{w}{r} \leq q \leq \dfrac{w}{\ell} 
 &\Longleftrightarrow  \left\lceil \dfrac{w}{r} \right\rceil \leq  \dfrac{w}{\ell} \nonumber
\end{align}
Therefore the sequences $s_j$, $j=1,..., q $ satisfying $\ell \leq  s_j \leq r $ exist if and only if $ \lceil \frac{w}{r} \rceil \leq  \frac{w}{\ell}$.
\end{proof}
%%%%%%%%%%%%%%%%%%%%%%%%%%%%%%%%%%%%%%%%%%%%%%%%%%%%%%%%%%%%%%%%%%%%%%%%%%%%%%%%%%%%%%%%%%%%%%%%%%%%%%%%%%%%%%%%%%%%%%%%%%%%%%%%%%%%%%%%%%%%%%%%%%%%%%%%%%%%%%%%%%%%%
\vspace{2mm}
\section{Proof of Lemma \ref{LemmaWclosest}}\label{ApxLemmaWclosest}
\begin{proof}
Instead of using $\sum_{i=1}^q s_i$, we just use its equivalent value $ q \cdot s' $, where $s' = \frac{1}{q} \sum_{i=1}^q s_i $. 
Suppose that there exists a value $v=q \cdot s'$ such that  $w^* < v < w $ which means $q \cdot s' >  h \cdot \lfloor \frac{w}{h} \rfloor $(Note that $q \cdot s' \neq w$ because $w$ is not achievable).  We know that: $s' \leq h $. Combining these two inequality results in: $ q> \lfloor \frac{w}{h} \rfloor$, which implies that:
\begin{align}\label{eqnW<1}
\left\lceil \dfrac{w}{h} \right\rceil \leq q
\end{align}
On the other hand $ \ell \leq s' $, so we have $ q \cdot \ell \leq q \cdot s' <w $ which results in:
\begin{align}\label{eqnW<2}
q < \dfrac{w}{\ell}
\end{align}

Combing the inequalities (\ref{eqnW<1}) and (\ref{eqnW<2}) results in $ \lceil \frac{w}{h} \rceil <  \frac{w}{\ell} $, which contradicts the fact that $w$ is not achievable ( Lemma \ref{LemmaWachievable}). Therefore the closest point to $w$ which is also smaller than $w$ is $w^* = h \cdot \lfloor \frac{w}{h} \rfloor  $.
%
%Similarly for the second part suppose that there exists a value $y=q \cdot s'$ such that $w < y <\tilde{w}$ which means: $ q \cdot s' < \ell \cdot \lceil \frac{w}{\ell} \rceil$.( $q \cdot s' \neq w$ because $w$ is not achievable).  We know that: $\ell \leq s' $, Combining these two inequality results in: $q < \lceil \frac{w}{\ell} \rceil $, which implies that: 
%\begin{align}\label{eqnW>1}
%q \leq  \dfrac{w}{\ell}
%\end{align}
%%
%On the other hand $s' \leq h   $, so we have $ w < q \cdot s' \leq q.h  $ which results in:
%\begin{align}\label{eqnW>2}
% \dfrac{w}{h} < q  \hspace{0.5cm} &\Longrightarrow \hspace{0.5cm}\left\lceil \dfrac{w}{h} \right\rceil \leq q
%\end{align}
%
%Combing the inequalities (\ref{eqnW>1}) and (\ref{eqnW>2}) results in $ \lceil \frac{w}{h} \rceil \leq  \frac{w}{\ell} $, which is again in contradiction with the fact that $w$ is not achievable.
% Therefore the closest point to $w$ which is also greater that $w$ is $\tilde{w} = \ell \cdot \lceil \frac{w}{\ell} \rceil $.
\end{proof}
%%%%%%%%%%%%%%%%%%%%%%%%%%%%%%%%%%%%%%%%%%%%%%%%%%%%%%%%%%%%%%%%%%%%%%%%%%%%%%%%%%%%%%%%%%%%%%%%%%%%%%%%%%%%%%%%%%%%%%%%%%%%%%%%%%%%%%%%%%%%%%%%%%%%%%%%%%%%%%%%%%%%%
\vspace{2mm}
\section{Proof of Lemma \ref{LemmaFillingOpt}}\label{ApxLemmaFillingOpt}
\begin{proof}
To prove Lemma \ref{LemmaFillingOpt} we will show that for any FSP policy $\theta$ with $P_{max}^{\theta}$, there exists a filling $\theta^*$ with $P_{max}^{\theta^*}$ such that $P_{max}^{\theta^*} \le P_{max}^{\theta}$.
In $\theta$, first consider an interval $\mathcal{T}=[t_1,t_2] $ such that $ P(t)= P_{max}^{\theta} ~~\forall t \in \mathcal{T} $. Now suppose there is a gap in $ \mathcal{T}$ such that the length of this gap is more than the length of a narrow rectangle completely placed in $\mathcal{T}$. Then this rectangle can be put in that gap while keeping the $ P(t)= P_{max}^{\theta}$ in $ \mathcal{T}$. Therefore we assume that 
%in $\theta$ there is no such intervals 
$ \mathcal{T}$ doesn't contain this kind of gaps. 
So in $\theta$ in every interval $\mathcal{T} $ such that $ P(t)= P_{max}^{\theta} ~~\forall t \in \mathcal{T} $ any rearrangement of the narrow rectangles 
%within each of the rows will 
does not increase $P_{max}^{\theta}$. Now we try to pick a rectangle scheduled in the interval $\mathcal{T} $ (partly or completely) and put it in an \emph{unfilled} row. 
%As we mentioned before this rescheduling will not increase the $P_{max} $. We will continue this procedure until every row becomes filled. It gives us a \emph{filling} $ \theta^*$ of narrow rectangle with $P_{max}=P_{max}^{\theta^*} \le P_{max}^{\theta} $. Therefore $P_{opt}^{F} = \min _{\theta^* \in \Theta^*} {P_{max}^{\theta^*}}$. 
This results in a new scheduling $\theta' $ with $P_{max}^{\theta'} \le  P_{max}^{\theta}$. Now we repeat this procedure for $\theta'$ and continue doing it until we get a scheduling $\theta^* $ in which every row becomes filled (maybe except for the last row), so $\theta^* $ is a filling with $P_{max}^{\theta^*} \le  P_{max}^{\theta}$.
\end{proof}
%%%%%%%%%%%%%%%%%%%%%%%%%%%%%%%%%%%%%%%%%%%%%%
\begin{figure}[tp]
\centering
\includegraphics[width=.50 \textwidth ] {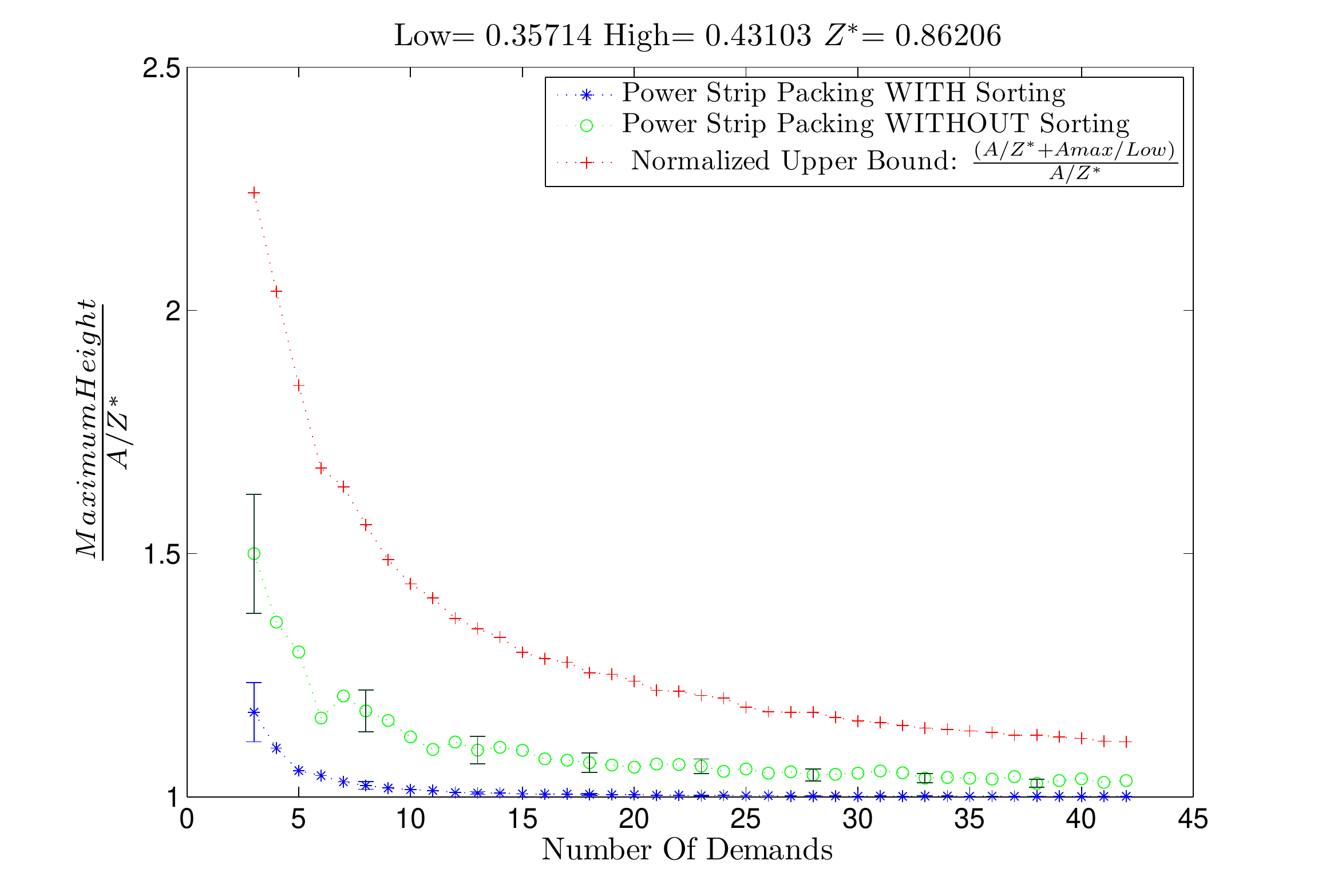}%\vspace{-1mm}
\caption{Comparing performances of the algorithms for Non-Ideal cases with  $ \frac{A}{Z^*} + \frac{A_{max}}{\ell} $ for different number of demands, where $\ell=0.3571 $ and $ r=0.43103  $ and $A_i$'s are  independent and identically uniformly distributed in the interval $ [0,\ell] $.}%\vspace{-1mm}
\label{fig:AlgNonIdeal}
\end{figure}
%%%%%%%%%%%%%%%%%%%%%%%%%%%%%%%%%%%%%%%%%%%%%%
\vspace{2mm}
\section{Proof of Lemma \ref{LemmaKrectangles}}\label{ApxLemmaKrectangles}
\begin{proof}
With respect to the conditions $ 0 < \ell \leq r < 1 $ and $\lceil \frac{1}{r} \rceil >  \frac{1}{\ell} $, we conclude that $t=1$ is not achievable. Now define $K_0 = \lfloor \frac{1}{r} \rfloor = \lfloor \frac{1}{\ell} \rfloor $.
%, and $ K_0 \cdot \ell = \left\lfloor \dfrac{1}{\ell} \right\rfloor \cdot \ell <1 $. 
So we have: $(K_0 + 1)= \lceil \frac{1}{\ell} \rceil > \frac{1}{\ell} $, which means: $ (K_0 + 1) \cdot \ell >1 $ and hence: $1- K_0 \cdot \ell < \ell$. 
The last inequality shows that 
%is an achievable value but 
%with respect to inequality (\ref{eqnLemmaK_l}) 
there is no space to add even the smallest narrow rectangle (i.e. $s_i=\ell$) to the point $ K_0 \cdot \ell \le 1 $ and since $\ell \leq s_i $, $K_0$ is the largest number of rectangles which can fill a row and don't exceed $t=1$. 
%
%\begin{align}\label{eqnLemmaK_l}
%$$(K_0 + 1)= \left\lceil \dfrac{1}{\ell} \right\rceil > \dfrac{1}{\ell} \Longrightarrow (K_0 + 1) \cdot \ell >1 \Longrightarrow 1- K_0 \cdot \ell < \ell $$
%\end{align}
%
On the other hand from Lemma \ref{LemmaWclosest}, $Z^* = K_0 \cdot r $ is the largest achievable value in $[0,1]$, i.e. $K_0 r \le 1 $ and $1- K_0 r \le \ell $. since $s_i \leq r$, $K_0$ is the smallest number of rectangles which can fill a row such that there is no space to add a new rectangle to this row.
Therefore $K_0 $ is the exact number of rectangles in each filled row.

\end{proof}
%%%%%%%%%%%%%%%%%%%%%%%%%%%%%%%%%%%%%%%%%%%%%%%%%%%%%%%%%%%%%%%%%%%%%%%%%%%%%%%%%%%%%%%%%%%%%%%%%%%%%%%%%%%%%%%%%%%%%%%%%%%%%%%%%%%%%%%%%%%%%%%%%%%%%%%%%%%%%%%%%%%%%
\begin{figure}[tp]
\centering
\includegraphics[width=.50 \textwidth ] {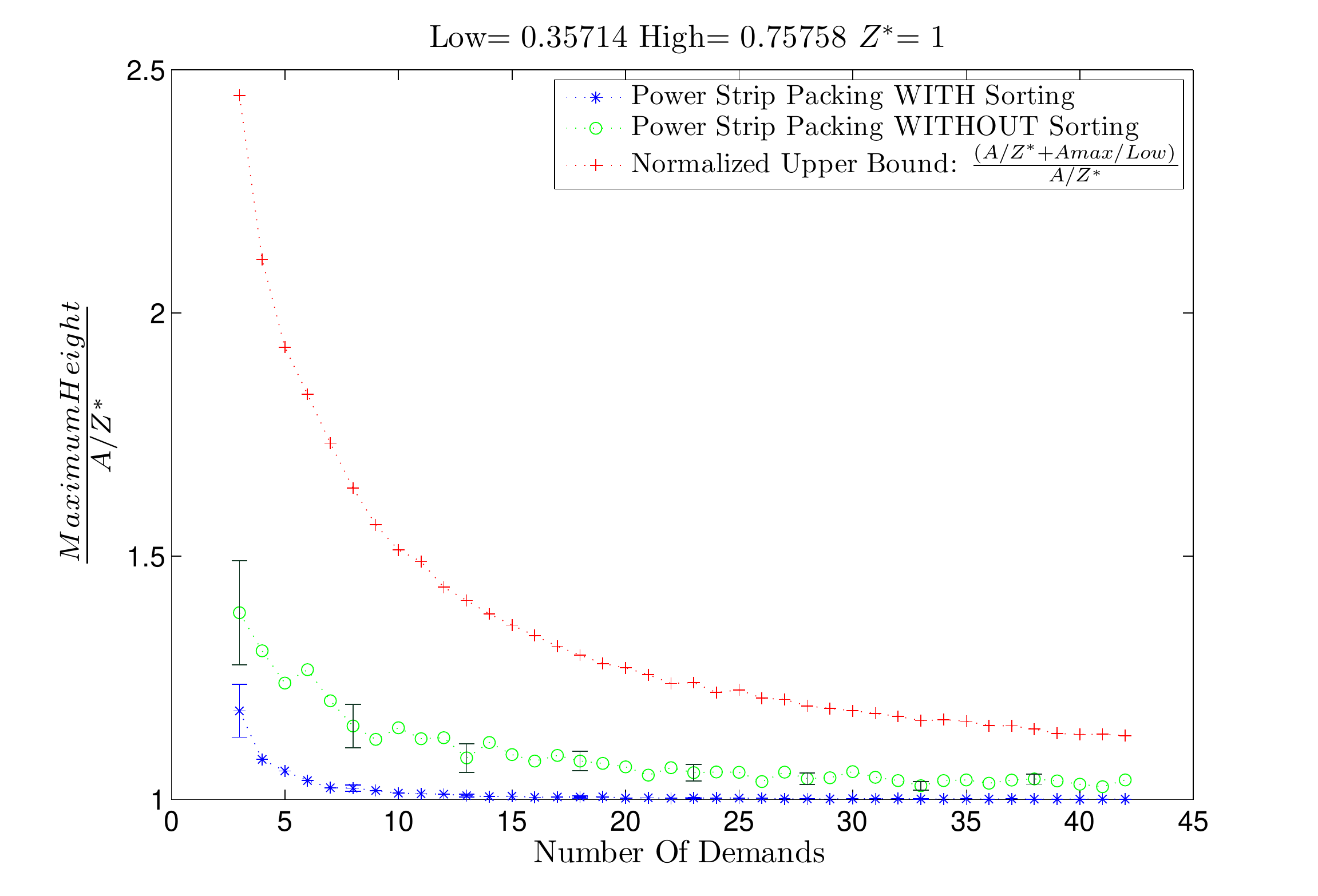}%\vspace{-1mm}
\caption{Comparing performances of the algorithms for Non-Ideal cases with  $ \frac{A}{Z^*} + \frac{A_{max}}{\ell} $ for different number of demands, where $\ell=0.35714 $ and $ r=0.75758  $ and $A_i$'s are  independent and identically uniformly distributed in the interval $ [0,\ell] $.}%\vspace{-3mm}
\label{fig:AlgNearIdeal}
\end{figure}

%%%%%%%%%%%%%%%%%%%%%%%%%%%%%%%%%%%%%%%%%%%%%%%%%%%%%%%%%%%%%%%%%%%%%%%%%%%%%%%%%%%%%%%%%%%%%%%%%%%%%%%%%%%%%%%%%%%%%%%%%%%%%%%%%%%%%%%%%%%%%%%%%%%%%%%%%%%%%%%%%%%%%
%\vspace{-1mm}
\bibliographystyle{IEEEtran}	% (uses file "plain.bst")
\bibliography{refs}		% expects file "myrefs.bib"

\begin{thebibliography}{10}
\providecommand{\url}[1]{#1}
\csname url@samestyle\endcsname
\providecommand{\newblock}{\relax}
\providecommand{\bibinfo}[2]{#2}
\providecommand{\BIBentrySTDinterwordspacing}{\spaceskip=0pt\relax}
\providecommand{\BIBentryALTinterwordstretchfactor}{4}
\providecommand{\BIBentryALTinterwordspacing}{\spaceskip=\fontdimen2\font plus
\BIBentryALTinterwordstretchfactor\fontdimen3\font minus
  \fontdimen4\font\relax}
\providecommand{\BIBforeignlanguage}[2]{{%
\expandafter\ifx\csname l@#1\endcsname\relax
\typeout{** WARNING: IEEEtran.bst: No hyphenation pattern has been}%
\typeout{** loaded for the language `#1'. Using the pattern for}%
\typeout{** the default language instead.}%
\else
\language=\csname l@#1\endcsname
\fi
#2}}
\providecommand{\BIBdecl}{\relax}
\BIBdecl

\bibitem{GridWise}
L.~D. Kannberg, M.~C. Kintner-Meyer, D.~P. Chassin, R.~G. Pratt, L.~A.
  DeSteese, J. G.and~Schienbein, S.~G. Hauser, and W.~M. Warwick,
  \emph{GridWise: The Benefits of a Transformed Energy System}.\hskip 1em plus
  0.5em minus 0.4em\relax Pacific Northwest National Laboratory under contract
  with the United States Department of Energy.
  http://arxiv.org/pdf/nlin/0409035, 2003.

\bibitem{GridFuture}
A.~Ipakchi and F.~Albuyeh, ``Grid of the future,'' \emph{IEEE Power and Energy
  Magazine}, vol. 7 , Issue 2, pp. 52--62, 2009.

\bibitem{ChargingStationConf}
I.~Bayram, G.~Michailidis, M.~Devetsikiotis, S.Bhattacharya, A.~Chakrabortty,
  and F.~Granelli, ``Local energy storage sizing in plug-in hybrid electric
  vehicle charging stations under blocking probability constraints,'' in
  \emph{IEEE SmartGridComm 2011 Track architectures and models}, 2011, pp. 78
  --83.

\bibitem{DLC}
N.~Ruiz, I.~Cobelo, and J.~Oyarzabal, ``A direct load control model for virtual
  power plant management,'' \emph{IEEE TRANSACTIONS ON POWER SYSTEMS}, vol. 24,
  No. 2, pp. 959--966, May 2009.

\bibitem{RTP}
C.~Triki and A.~Violi, ``Dynamic pricing of electricity in retail markets,''
  \emph{A QUARTERLY JOURNAL OF OPERATIONS RESEARCH}, vol. 7, No. 1, pp. 21--36,
  March 2009.

\bibitem{RadWung2010}
A.~H. Mohsenian-Rad, V.~W.~S. Wong, J.~Jatskevich, R.~Schober, and
  A.~Leon-Garcia, ``Autonomous demand-side management based on game-theoretic
  energy consumption scheduling for the future smart grid,'' \emph{IEEE
  TRANSACTIONS ON SMART GRID}, vol. 1 , Issue 3, no.~2, pp. 320--331, 2010.

\bibitem{Incentive}
S.~Caron and G.~Kesidis, ``Incentive-based energy consumption scheduling
  algorithms for the smart grid,'' in \emph{First IEEE International Conference
  on Smart Grid Communications (SmartGridComm)}, 2010.

\bibitem{Tassiulas}
\BIBentryALTinterwordspacing
I.~Koutsopoulos and L.~Tassiulas, ``Control and optimization meet the smart
  power grid: Scheduling of power demands for optimal energy management,'' in
  \emph{Arxiv preprint: http://arxiv.org/abs/1008.3614v1}, 2010. [Online].
  Available: \url{http://arxiv.org/abs/1008.3614v1}
\BIBentrySTDinterwordspacing

\bibitem{ElasticDemands}
\BIBentryALTinterwordspacing
J.~Le~Boudec and D.~Tomozei, ``Satisfiability of elastic demand in the smart
  grid,'' in \emph{Arxiv preprint: http://arxiv.org/abs/1011.5606v2}, 2011.
  [Online]. Available: \url{http://arxiv.org/abs/1011.5606v2}
\BIBentrySTDinterwordspacing

\bibitem{Coeffman1980}
E.~Coffman, M.~Garey, D.~Johnson, and R.~Tarjan, ``Performance bounds for
  level-oriented two-dimensional packing algorithms,'' \emph{SIAM Journal on
  Computing}, vol. 9, Issue 4, pp. 808--826, 1980.

\bibitem{kenyon2000}
C.~Kenyon and E.~R{\'e}mila, ``{A near-optimal solution to a two-dimensional
  cutting stock problem},'' \emph{Mathematics of Operations Research}, vol. 25,
  No. 4, pp. 645--656, 2000.

\bibitem{lodi2002}
A.~Lodi, S.~Martello, and M.~Monaci, ``{Two-dimensional packing problems: A
  survey},'' \emph{European Journal of Operational Research}, vol. 141, Issue
  2, no.~2, pp. 241--252, 2002.

\bibitem{MalleableJansen}
K.~Jansen, ``Scheduling malleable parallel tasks: an asymptotic fully
  polynomial time approximation scheme,'' \emph{Algorithmica}, vol. 39, Number
  1, pp. 59 -- 81, 2004.

\bibitem{FSP}
N.~Karmarkar and R.~Karp, ``An efficient approximation scheme for the
  one-dimensional bin-packing problem,'' in \emph{23rd Annual Symposium on
  Foundations of Computer Science (SFCS '08)}, 1982.

\end{thebibliography}
%\bibliography{"C:/MMK/Dropbox/SmartGrid/LastRef/refs"}		% expects file "myrefs.bib"
%\bibliography{refs}		% expects file "myrefs.bib"

\end{document}